%% file: main.tex
\documentclass{llncs}

\pdfoutput=1 

\usepackage{amssymb}
\usepackage[inline]{enumitem}
\usepackage{multirow}
\usepackage[dvipsnames]{xcolor}
\usepackage{extarrows}
\usepackage{bm}
\usepackage{subfigure}
\usepackage{tikz}
\usepackage{url}
\usetikzlibrary{arrows}
\usepackage[Algorithmus]{algorithm}
\usepackage[linktoc=page]{hyperref}
\usepackage{algorithmic}
\usepackage{fancyhdr}

\hypersetup{
	colorlinks,
	linkcolor={red},
	citecolor={blue},
	urlcolor={OliveGreen}
}
\usepackage[capitalise]{cleveref}

\newtheorem{assumption}{Assumption}
\begin{document}

\title{Optimally Integrating Ad Auction into E-Commerce Platforms}

\author{
Weian Li  \inst{1} \and
Qi Qi \inst{2} \and
Changjun Wang \inst{3} \and
Changyuan Yu \inst{4}
}

\institute{Center on Frontiers of Computing Studies, Peking University, Beijing, China\\
\email{weian\underline{~}li@pku.edu,cn} \and
Gaoling School of Artificial Intelligence, Renmin University of China, Beijing, China\\
\email{qi.qi@ruc.edu.cn}\and
Academy of Mathematics and Systems Science, Chinese Academy of Sciences, Beijing, China\\
\email{wcj@amss.ac.cn} \and
Baidu Inc., Beijing, China\\
\email{yuchangyuan@baidu.com}
}

\maketitle


\begin{abstract}
    Advertising becomes one of the most popular ways of monetizing an online transaction platform. Usually, sponsored advertisements are posted on the most attractive positions to enhance the number of clicks. However, multiple e-commerce platforms are aware that this action may hurt the search experience of users, even though it can bring more incomes. To balance the advertising revenue and the user experience loss caused by advertisements, most e-commerce platforms choose fixing some areas for advertisements and adopting some simple restrictions on the number of ads, such as a fixed number $K$ of ads on the top positions or one advertisement for every $N$ organic searched results. Different from these common rules of treating the allocation of ads separately (from the arrangements of the organic searched items), in this work we build up an integrated system with mixed arrangements of advertisements and organic items. We focus on the design of truthful mechanisms to properly list the advertisements and organic items and optimally trade off the instant revenue and the user experience. Furthermore, for different settings and practical requirements, we extend our optimal truthful allocation mechanisms to cater for these realistic conditions. Finally, 
    we exert several experiments to verify the improvement of our mechanism compared to the common-used advertising mechanism.
\end{abstract}

\keywords{E-Commerce Platforms \and Ad Auctions \and Optimal Mixed Arrangements.}

\setcounter{footnote}{1}
\input{Source/Intro}

\input{Source/Pre}

\input{Source/Opt}

\input{Source/Ext}

\input{Source/Exp}

\section{Conclusion}
\label{sec:conclusion}
To sum up, for the new features of e-commerce online platforms that combine advertisement and organic items for display purposes, we offer a thorough and exact study on designing mechanisms to balance both areas of ``volume'' and revenue. 
For different ways of trade-off and practical requirements, optimal feasible mechanisms can be designed in order to realize them. We also empirically evaluate our mechanism and demonstrate its advantage. 

As far as we know, our research is the first through study on mixed arrangement of advertisement items and organic items. There are still many problems remain open, to name a few:
\begin{itemize}
\item In this paper, we only focus on feasible mechanisms. There may exist mechanisms which are not limed to BIC or IR that can realize a higher trade-off.
\item The valuation distribution is publicly known to each advertiser. A further area to explore would be how to design mechanisms that avoid cheating on distribution.
\end{itemize}

\newpage
\bibliographystyle{splncs04}
\bibliography{adallocation}

\newpage
\setcounter{page}{1}
\input{Source/Appendix}

\end{document}

%% file: Source/Intro.tex
\section{Introduction}
With the development of Internet and mobile devices, e-commerce platforms have become the most fashionable online marketplace to bridge merchants and consumers, whose popularity is originated from its convenient shopping mode. That is, whenever and wherever consumers want to pick up their desirable products, they only require to enter a relevant keyword into the search bar of Apps of e-commerce platforms by mobile devices. Then, platforms will return a list of the matching products in the search result pages provided to be selected. Hence, how to recommend products accurately and effectively is a considerable problem of platforms, which also affects both the short-run revenue and the long-run prosperity.

\setcounter{page}{1}

Nowadays, in one standard webpage of search results, based on a given keyword, it usually presents two types of items, \emph{organic search results} and \emph{sponsored advertisements\footnote[1]{Sponsored advertisements are usually labeled by ``Sponsored'' or ``Ad'' to be distinguished with organic search results.}}. 
The platform calculates the displayed sequence of two types of items followed by certain rules. Traditionally, organic search results and sponsored advertisements are totally separate components in the process of arranging items. First, e-commerce platforms will decide how many top positions are used to display advertisements, and then run a preset  auction mechanism, like generalized second price auction (GSP), to output the ads. For the rest positions, they are used to show the organic items ranked by the relevance of items or other criteria which depend on different platforms (e.g., see \cite{AL04,AU06,WH10}).

The first step of the above process is so-called \emph{Sponsored Search Auction}, which is first launched by Google at the end of 20th century. Several common pricing models can be adopted in sponsored search auctions, like pay-per-click, pay-per-impression and pay-per-action. In this paper, we focus on the pay-per-click model which means that, for the winners of auction, they only pay when their advertisements are clicked. 
Due to the dissatisfaction with the ranking of organic items, some merchants hope to join the auction (becoming the advertisers) and charge some extra money for improving the position and drawing more attention from customers. Because of the millions of search everyday, sponsored search auction has become the vital instant incoming source of e-commerce platform today. Different from the first step, ranking organic items aims to enhance the efficiency of search and user experience, so its criterion is usually the relevance of items. In this investigation, we consider exploiting a relatively comprehensive but very simple indicator to reflect relevance, that is, the expected merchandise volume, or the expected sale amount to rank the organic items with two reasons. The first reason is that the expected merchandise volume includes the user's interest in browsing the product (click rate), the detailed purchase tendency (conversion rate), and the decision after the balance between price and product quality. In another word, the merchandise volume can directly or indirectly reflect the user experience. The second reason is that, nowadays, most e-commerce platforms regard the gross merchandise volume (GMV) as a criterion to measure how much they take over the whole market, and compete for GMV during some shopping ceremonies.

However, due to some specialty of e-commerce platforms, like that the sponsored advertisements are also the organic items essentially, still insisting on the conventional pattern and regarding sponsored advertisements and organic results as independent sections may expose some downsides:
\begin{enumerate*}[label = (\alph*), font = {\bfseries}]
\item {obviously, it is not appropriate for different keywords to use the same number of positions providing for ads. The conventional pattern cannot automatically solve the above problem;}

\item {to optimize the user experience, the platform will list the organic results according to their expected transaction volume or their expected sale quantities (these volumes or quantities are well known by the platforms), i.e., the item with the highest expected ``volume''  will be shown in the first slot (exclude the slots for advertisements), the item with the second highest expected ``volume'' will be shown in the second slot, and so on. Nevertheless,  during the sponsored search auction, if some items with low volume would rather pay high price to get more clicks, by the goal of sponsored search auction, maximizing the instant revenue, these items will be picked up, which will harm the platforms' expected GMV and the long-run profits.}
\end{enumerate*}
In view of the above points, is it still a great idea to fix the ad positions in advance? Can the platforms do better if they weighing the allocation of the sponsored advertisements and the organic results at the same time, not separately? If yes\footnote{We give an example to illustrate that the GMV and revenue will be better if we consider organic results and sponsored advertisements at the same time in appendix}, how to balance the two conflicting objectives of user experience (i.e., GMV) and instant advertising revenue in designing mechanisms? Furthermore, how to design the optimal mechanisms with the integrated allocation?

In reality, on some platforms such as Tmall, Taobao etc, the advertisements and organic items have already been combined together to display. However, the currently used mixed allocation rules are still very heuristic. For example, always take out one slot for advertisement of every $N$ displaying slots. Motivated by this, in this work, we will, from the theoretical aspect, try to study how to design truthful mechanisms to optimally trade off the expected revenue and volume for this new integrated setting (answering the third question and fourth question proposed above). To the best of our knowledge, our work is the first attempt on studying auction mechanism design with multiple objectives in the new layout of mixed arrangements.

\subsection{Our Contribution}
In this paper, our contributions can be summarized as follows.
\begin{itemize}
\item \textbf{Novelty}---We initially build up an integrated model to describe the layout of mixed arrangements of organic results and sponsored advertisements, called an integrated ad system (IAS), where two main objectives, revenue and GMV, are considered. 
    We propose two general kinds of problems with two different ways to trade off the revenue and GMV: unconstrained  problem of linearly combining the volume and revenue as a single objective; constrained problem of bounding the volume and maximizing the revenue (see Section \ref{CorePro}). 
\item \textbf{Techniques}---For the unconstrained  problem,  we prove that the optimal truthful mechanism (Theorem \ref{thm1}) can be obtained by a transformation on objectives and then allocating the items by their ``revised virtual values''. For the constrained problem, using variables relaxation and Lagrangian dual methods,  we show that the constrained problem can be  transformed into an unconstrained problem by properly choosing a parameter (Theorem \ref{thm:rlt}). Then we give a numerical algorithm (Algorithm \ref{alg:algorithm}) to compute the desired parameter, thus deriving the optimal truthful mechanism for the constrained problem.
\item \textbf{Extensions}---We extend our study to three general settings with practical restrictions: requiring an upper bound of the total number of advertisements, and requiring certain sparsity of the allocation of advertisements where for the later, we divide it into two models with different restriction on the sparsity. However, under these restrictions, some good properties will not exist anymore, but by some more delicate treatments, we can still  design the optimal mechanisms for both constrained and unconstrained problems.
\item \textbf{Practicability}---We first verify that our mechanisms can be implemented as planned. Then comparing with the currently used mechanism, we show the superiority of our mechanisms. At last, we take the correlation between value and weight of advertisers into account, our mechanisms still perform better.
\end{itemize}

\subsection{Related Work}

\noindent\textbf{Sponsored Search Auction.}
Auction lies in the core of mechanism design research, while sponsored search auction has been an area of great focus in computer science in the last twenty years.  For the traditional revenue optimization, \cite{Myerson} solved the problem for single item in Bayesian-Nash equilibrium setting. \cite{Mask} studied multi-unit optimal auction. From the first launch of sponsored search auction in 1997, a series of auction mechanisms are put forward, like generalized first price(GFP) auction by Overture and generalized second price(GSP) by Google, which makes sponsored search auction become a vital incoming source of variable online platforms in practise. \cite{GAR06} showed that the lower bound of revenue produced by GSP, the most popular mechanism is equal to the revenue of classic VCG mechanism (\cite{VK61,CLK71,GRV73}). Due to the untruthfulness of GSP, \cite{EOS07} and \cite{VR07} independently examined the behavior of bidders and proposed locally envy-free equilibrium(LEFE) and symmetric Nash equilibrium (SNE), respectively. Another idea of a squashing parameter of GSP to improve the revenue was proposed by \cite{Lah}.\cite{Ostr} applied these works and studied its effects on Yahoo! auction using the optimal reserve price. \cite{Thom} also studied several different techniques for increasing the revenue, including via several different reserve prices, squashing, combinations of a reserve price and squashing, ect. Rather than using the reserve price simply as a minimum bid, \cite{Rob} presented the idea of incorporating the reserve price into ranking score and showed that this mechanism may increase the revenue compared to the squashing mechanism.

For the topic of user experience, \cite{Abr} introduced a concept of hidden cost, which was advertiser-specific and represented the quality of the ads' landing page. \cite{Ath} studied a model that introduced the users' search cost, and showed that reserve price can improve users' welfare. \cite{Li} defined the shadow cost, which is revenue reduction in long run and depends on both advertiser and slot.

As for the trade-offs among different objectives, there are also some related studies. \cite{Lik} first studied the problem of designing optimal mechanisms to balance revenue and welfare. \cite{Sun} considered the convex combination of revenue and welfare to improve the prediction. \cite{Bac} applied the linear combination of different objectives as their objective function. \cite{Shen} introduced a class of parameterized mechanisms to balance different objectives.

\noindent\textbf{Rank of organic results.}
Besides the methods (\cite{AL04,AU06,WH10}) used in practise, there are still multiple theoretical investigations about ranking organic results, recently. \cite{CRW07} first incorporated bias of search results into consideration, when ranking the items. In the next few years, this topic is still studied by many researchers, e.g., \cite{EL11,WR12,MT14}. \cite{LMST17} showed how to optimally rank search results to maximize an objective that combines search-result relevance and sales revenue. \cite{CNZ17} studied the optimal ranking rule of  multiple objectives includes consumer and seller surplus, as well as the sales revenue, taking consumers' choice into consideration.

However, all these work treats the sponsored search auction separately and does not consider the mixed arrangement of organic results and advertisements. We design an integrated system that takes both organic items and paid advertisements into consideration together.

\noindent\textbf{Our organizations.}
In Section \ref{sec:pre}, we give all the necessary notations and definitions, and build up the modelling framework of the integrated ad system. Two core optimization problems are also proposed in this section. In Section \ref{sec:opt}, we consider designing optimal mechanisms for these two problems. In Section \ref{sec:ext}, we generalize our study to more practical settings and design the corresponding optimal mechanisms. In Section \ref{sec:exp}, we run several numerical experiments to testify the performance of our mechanisms. Finally, we summarize our works and put forward open problems in Section \ref{sec:conclusion}. 

%% file: Source/Pre.tex
\section{Notations and Preliminaries}
\label{sec:pre}
\subsection{Integrated Ad System}\label{pre-IAAS}
In the classic ad auction,  advertisers bid for the keywords, and the pre-set advertising slots, which are separated from organic search result, are allocated to the ads with the highest bids. In this paper, we propose a new model called \emph{integrated ad system} (IAS) which caters for modern e-commerce online platforms, where both the positions and total number of advertising slots are not determined in advance. After collecting bids from advertisers, the platform decides how to rank all items (including sponsored advertisements and organic results) and how much each advertiser should pay.

We start with a formal description of the IAS. In an IAS, there are $n_1$ advertisers and $n_2$ organic results competing for $K$ available slots of one search-result page simultaneously. 
For simplicity, we call advertisements and organic results as ad items and organic items, respectively. Denote $A$ as the set of ad items and $O$ as the set of organic items, and exploit $i \in A$ or $i \in O$ to represent an ad item or an organic item. Let $k \in \{1, 2, \ldots, K\}$ index the slots. Generally, a higher position slot has a smaller index. For the $k$th slot, $\beta_k$ stands for its effective exposure, which means the probability that one user pays attention to this slot. Without lose of generality, assume that $\beta_1 > \beta_2 >\ldots> \beta_K >0$.

Each (ad or organic) item $i$ has a quality (weight) factor $w_i$ to reflect its relative popularity compared to other items. In addition, each item $i$'s estimate merchandise volume per click is $g_i$ that is well known by the platform and advertisers. In detail, the estimate merchandise volume is an attribute that represents the expected sale amount per click. For ease of representation, we call $g_i$ volume, instead. Each ad item $i$ still has an extra private value $v_i$ per click. More specifically, $v_i$ means that the advertisers want to charge extra $v_i$ to gain one click. Suppose that $v_i$ is independently (not necessarily identical) drawn from $[0, u_i]$ according to a publicly known distribution $F_i(v_i)$ whose respective pdf is $f_i(v_i)$. Given this, the virtual value of ad item $i$ can be represented as $\phi_i(v_i)=v_i-\big(1-F_i(v_i)\big)/f_i(v_i)$. We assume that the distribution $F_i(v_i)$ satisfies the regular condition, which implies that $\phi_i(v_i)$ is monotone non-decreasing. For the organic items, we have the following assumption:
\begin{assumption}
For any organic item $i$, we assume its valuation is always 0, i.e., $v_i$ is drawn from the degenerate distribution with one support point 0. \footnote{This assumption can be understood as that any organic items have no intention to charge extra money for click numbers. They do not submit any price to platform and make no contribution to the instant revenue.}.
\end{assumption}
In the IAS, we consider the separable click-through-rate(CTR) model. That is, if the item $i$ is allocated to slot $k$, the item $i$'s CTR is $w_i\beta_k$. It means that if we put an item into a slot, the probability of being clicked is affected not only by the position but also by itself. Similarly, the corresponding merchandise volume is $g_iw_i\beta_k$.

For convenience, we summarize the above symbols in Table \ref{SYB} for easily reading. In summary, the whole process of IAS can be described as: 1. The platform releases the information of slots; 2. The advertisers submit a bid price based on the value; 3. The platform ranks all ad items and organic items simultaneously and decides how much the advertisers should charge.
\begin{table}[ht!]
    \begin{center}
        \begin{tabular}{l | l}
            \hline
            $n_1$ & the number of advertisers \\ \hline
            $n_2$ & the number of organic items \\ \hline
            $i$   & one advertisement or one organic item \\ \hline
            $k$   & one slot \\ \hline
            $w_i$ & the weight of ad or organic item $i$ \\ \hline
            $v_i$ & the value of ad or organic item $i$ \\ \hline
            $F_i(v_i)$ & the cdf of $v_i$ \\ \hline
            $f_i(v_i)$ & the pdf of $v_i$ \\ \hline
            $\phi_i(v_i)$ & the virtual value of ad or organic item $i$ \\ \hline
            $g_i$ & the estimate merchandise volume per click of ad or organic item $i$ \\ \hline
            $\beta_k$ & the effective exposure of slot $k$ \\ \hline
        \end{tabular}
        \caption{The Symbols of Integrated Ad System}
        \label{SYB}
    \end{center}
\end{table}

\subsection{Mechanism Design}\label{pre-md}
Let $v=(v_1, v_2,\ldots, v_{n_1})$ and $b=(b_1, b_2, \ldots, b_{n_1})$ be the value profile and the bid profile of all ad items.  We may use $v_{-i}$ and $b_{-i}$ to represent the value profile and bid profile of all ad items except ad item $i$.

In the IAS, after receiving the bids from advertisers, a \emph{mechanism} $\mathcal{M}=(x(b),p(b))$ consists of two rules, allocation rule $x(b)$ and payment rule $p(b)$. More specifically, $x(b)=(x_1(b), x_2(b), \ldots,$ $ x_{n_1+n_2}(b))$\footnote{In the IAS, since we allow that ad items and organic items display in a mixture configuration, the outcome of an mechanism decides how to allocate both two types of items, simultaneously.} and $p(b)=(p_1(b), p_2(b), $ $\ldots, p_{n_1}(b))$, where $x_i(b)=\sum_{k=1}^K x_{ik}(b)\beta_{k}$ and $x_{ik}(b)$ is the indicator function to imply whether item $i$ is assigned to slot $k$. Note that the CTR of item $i$ is $w_i x_i(b)$.

Since one slot should be allocated to one item and one item should be assigned to at most one slot, $x_{ik}(b)$ must satisfy the following constraints, denoted by $\mathcal{X}$:
\begin{align*}
&\sum_{k}x_{ik}(b)\leq 1, &\forall i\in A ~\text{or}~  O,\\
&\sum_{i\in A}x_{ik}(b)+\sum_{i\in O}x_{ik}(b) =1,  &\forall k,\\
&x_{ik}(b)\in \{0,1\}, &\forall i, k.
\end{align*}

Given allocation rule and payment rule, the expected utility of ad item $i$ with bid $b_i$, can be expressed as:
\begin{equation}
U_i(x,p,b_i)=\int_{V_{-i}}\big(v_i-p_i(b_i, v_{-i})\big) w_i x_i(b_i, v_{-i})f_{-i}(v_{-i})dv_{-i}
\end{equation}

In the IAS, to guarantee that advertisers have a desire to participate in the auction, we require that the utility of ad item should not be less than zero.

\begin{definition}[Individual Rationality]
\begin{equation*}
U_i(x,p,b_i) \geq  0, \quad\quad \forall b_i\in [0,u_i], \quad \forall i\in A.
\end{equation*}
\end{definition}

In this paper, we also hope to design mechanisms that avoid advertisers false reporting. Since our model is in Bayesian setting, we consider the mechanisms with Bayesian Incentive Compatibility (BIC).
\begin{definition}[Bayesian Incentive Compatibility]
\begin{equation*}
U_i(x,p,v_i)\geq U_i(x,p,b_i), \quad\quad \forall b_i\in [0, u_i], \quad \forall i\in A.
\end{equation*}
\end{definition}
We call a mechanism \emph{feasible} iff it is IR, BIC and satisfies the condition $\mathcal{X}$. 
Fortunately, \cite{Myerson} has already given the equivalent characterization of IR, BIC mechanisms. With a little refinement, we can design mechanisms with IR and BIC in IAS as the following.
\begin{lemma}[Myerson 1981]\label{lem-IC}
A mechanism is IR, BIC if and only if, for any ad item $i$ and bids of other items $b_{-i}$ fixed,\\
1. $x_i(b_i,b_{-i})$ is monotone non-decreasing on $b_i$.\\
2. $p_i(b)=b_i-(\int_{0}^{b_i}x_i(s_i ,b_{-i})ds_i)/x_i(b)$, when $x_i(b)\neq 0$; Otherwise, $p_i(b)= 0$.
\end{lemma}

In the following of this paper, we only concentrate on designing feasible mechanisms and use $v_i$ to represent the bid price $b_i$ directly for convenience.

\subsection{Core Problems}
\label{CorePro}
As mentioned before, e-commerce platforms, unlike conventional search engines, are concerned about instant revenue from ad items and GMV from both ad items and organic items. Recall the notations in Subsection \ref{pre-IAAS} and \ref{pre-md}, the revenue and GMV are respectively given by
\begin{equation}\label{rev}
\text{Revenue} = \int_V \sum_{i\in A} p_i(v) w_i x_i(v) f(v)  dv
\end{equation}
and
\begin{equation}\label{vol}
\text{GMV} = \int_V \Big[\sum_{i\in A} g_i w_i x_i(v)+\sum_{i \in O} g_i w_i x_i(v)\Big] f(v)  dv.
\end{equation}

In this subsection, we put forward two different approaches to trade off the two objectives. The first approach is to think about the linear convex combination of the revenue and the GMV directly. Specifically, given a coefficient $\alpha \in [0,1]$ in advance, our problem is to find an optimal mechanism to maximize their convex combination with $\alpha$. We call this problem as \emph{Unconstrained Problem}.
\paragraph{Unconstrained Problem $UCST(\alpha)$}
Given the weighted coefficient $\alpha \in [0,1]$, the unconstrained problem $UCST(\alpha)$ can be written as
\begin{equation}{\label{defi-ucst}}
\max \limits_{x\in\mathcal{X}} \quad \alpha \cdot \text{Revenue} +(1-\alpha)\cdot \text{Volume}.
\end{equation}

The second approach is that we optimize one metric while restricting the other metric. This problem is called \emph{Constrained Problem}. In this paper, we always put GMV into constraints and optimize the revenue. Hence, in this problem, a threshold of GMV, denoted by $V_0$, is always given in advance.
\paragraph{Constrained Problem $CST(V_0)$}
Given a threshold $V_0$, the constrained problem $CST(V_0)$ can be written as
\begin{align*}
\label{cp-prog}\tag{P1}
\max  & \quad\quad \text{Revenue}  \\
\tag{C1.1}
\rm{s.t.}  & \quad\quad \text{Volume} \geq V_0\\
\label{cp-c1.2}\tag{C1.2}
 & \quad \quad x\in \mathcal{X}
\end{align*}

%% file: Source/Opt.tex
\section{The Optimal Mechanisms for the IAS}
\label{sec:opt}
In this section, we mainly explore the optimal feasible mechanisms for both unconstrained problem and constrained problem. We start with the unconstrained problem and present the elegant optimal mechanism. 

\subsection{The Unconstrained Problem}
\label{sec:ust}

From our intuition, if there is no ad item, i.e., nobody submits a bid, the platform will rank organic items by their volumes naturally. 
Because the organic item will not influence the revenue, even if ad items appear, the order among organic items does not change, as well. Therefore, the problem of designing the optimal mechanism is, in fact, how to pick up appropriate ad items and insert them properly into the list of organic items to achieve the optimum of goals. More specifically, we need to design the optimal ranking criterion to incorporate the ad items without disturbing the order of organic items, and the corresponding payment function. For simplicity, we can define $I= A \cup O$ as the new set to represent all items. Formally, we also extend $v, b, p(b)$ to $n_1+n_2$ dimensions.

For the unconstrained problem with fully mixed allocation of the ads and organic items, we aim to find the criterion of rank to insert ad items into the sequence of organic items optimally. Because there is no other constraint except for the feasibility, we can simply the objective function to an easily observed form which can help us find the optimal ranking rule, shown in the Lemma \ref{lem1}.
\begin{lemma}\label{lem1}
To maximize the objective function (\ref{defi-ucst}), the mechanism needs to maximize the objective
\begin{equation}\label{eqn-ucst}
\int_V \sum_{i\in I} (\alpha \phi_i(v_i)+(1-\alpha)g_i) w_i x_i(v) f(v) dv-\alpha\sum_{i\in I}  U_i(x,p,0).
\end{equation}
\end{lemma}

\begin{proof}
Plugging (\ref{rev}) and (\ref{vol}) into (\ref{defi-ucst}), we can rewrite the objective function as
\begin{align}
\text{OBJ} &= \alpha\cdot \text{Revenue} +(1-\alpha) \text{Volume} \nonumber \\
    &\overset{(\dag)}{=} \alpha\left(\int_V \sum_{i\in I} p_i(v) w_i x_i(v) f(v)  dv\right)+(1-\alpha)\int_V \sum_{i\in I} g_i w_i x_i(v)f(v)  dv \nonumber \\
    &\overset{(\ddag)}{=} \alpha\left[-\sum_{i\in I}  U_i(x,p,0)+\int_V\sum_{i\in I} \phi_i(v_i) w_i x_i(v)f(v)  dv\right] \nonumber\\ &+(1-\alpha)\int_V \sum_{i\in I} g_i w_i x_i(v)f(v)  dv \nonumber \\
    &= \int_V\sum_{i\in I} w_i (\alpha\phi_i(v_i)+(1-\alpha)g_i) x_i(v)f(v)  dv -\alpha\sum_{i\in I}  U_i(x,p,0).
\end{align}
where $(\dag)$ follows from that organic items pay nothing to platforms, and $(\ddag)$ is the previous conclusion from \cite{Myerson}.
\end{proof}

Because $p(b)$ only appear in $U_i(x,p,0)$ and we need to guarantee the property of IR, by Lemma \ref{lem-IC}, if we choose
\begin{equation}\label{pym-rule}
p_i(v)=v_i-\frac{\int_{0}^{v_i}x_i(s_i ,v_{-i})ds_i}{x_i(v)}
\end{equation}
as payment rule, then $U_i(x,p,0)$ will be 0 and we only need to consider to select appropriate $x(b)$ to maximize the first part of (\ref{eqn-ucst}).

As for allocation rule, we define $\psi_i(v_i) \overset{def}{=} (\alpha\phi_i(v_i)+(1-\alpha)g_i) w_i$ as the revised virtual value of item $i$. It is not difficult to check that the following allocation rule can maximize the first part of formula (\ref{eqn-ucst}).
\begin{equation}\label{alc-rule}
x_{ik}(v)=
\begin{cases}
1 & \text{if } \psi_i(v_i) \text{ is the } k\text{th-highest revised virtual value}, k\in\{1,2,\ldots, K\}, \\
0 & \text{otherwise.}
\end{cases}
\end{equation}
Note that since $\phi_i(v_i)$ is regular, $\psi_i(v_i)$ is non-decreasing as well. Consequently, $x_i(v)$ satisfies the non-decreasing, when fix $v_{-i}$.


By the argument above, we can obtain the optimal mechanism for unconstrained problem by Theorem \ref{thm1}.
\begin{theorem}\label{thm1}
Given $\alpha\in [0,1]$, the mechanism $\mathcal{M}=(x(v), p(v))$, where $x(v)$ and $p(v)$ are described by (\ref{alc-rule}) and (\ref{pym-rule}), is the optimal feasible mechanism for the unconstrained problem $UCST(\alpha)$.
\end{theorem}

\begin{proof}
By Lemma \ref{lem-IC}, $x_i(b)$ and $p_i(b)$ decide the feasibility of this mechanism. Meanwhile, for any bid profile $b$, $x_i(b)$ and $p_i(b)$ also achieve maximum of objective, due to Lemma \ref{lem1}. Therefore, this mechanism is the optimal feasible mechanism.
\end{proof}


The optimal mechanism of the unconstrained problem has been given above. We will explain more about the optimal mechanism. For the special case that all items are organic items, the platform should rank them by their volumes. 
Our mechanism exactly has the same rank as this, because our mechanism ranks items by $\psi_i(v_i) = (\alpha\phi_i(v_i)+(1-\alpha)g_i) w_i= (1-\alpha)g_i w_i$. It is equivalent to ranking them by their volume. Based on this point of view, we can regard the rank of general case in another perspective. For ad items, the term $\phi_i(v_i)$ is different from that of organic item. In the same ranking criterion, we can rank the organic items first, then insert ad items optimally by ranking score $\psi_i(v_i)$. Another fact is that there exists no redundant slots that no items are assigned into. This is because, for the organic items, their revised virtual value is always greater than zero, which makes that, at least, the organic items can be shown in the result list. This point is also the difference with the classic Myerson auction which may not allocate items. In another words, our mechanism is based on the structure of Myerson auction, but allocation result is different. Finally, for the organic items, they also have the payment rule $p(b)$, but we can find that this value will be always equal to 0, which meets the practical requirement.

\subsection{The Constrained Problem}
\label{sec:cst}
In this subsection, we mainly concentrate on the constrained problem. Since the feasible region is discrete, i.e., $x_{ik}(v)$ is not a continuous function about bids $v$, it is difficult to optimize the objective directly. We first relax constraints $x_{ik}(v)\in  \{0, 1\}$ in (\ref{cp-c1.2}) to $x_{ik}(v)\in [0, 1]$ and denote the relaxed feasible region as $\bar{\mathcal{X}}$. 
The relaxed constrained problem is
\begin{align*}
\label{rcp-prog}\tag{P2}
\max  & \quad\quad R(x)\overset{def}{=}\int_V \sum_{i\in I} w_i p_i(v) x_i(v)f(v) dv  \\
\label{rcp-c2.1}\tag{C2.1}
\rm{s.t.}  & \quad\quad V(x)\overset{def}{=}\int_V \sum_{i\in I}  g_i w_i x_i(v)f(v)dv \geq V_0 \\
\label{rcp-c2.2}\tag{C2.2}
 & \quad\quad x\in \bar{\mathcal{X}}.
\end{align*}
Given this, we can prove that Program (\ref{rcp-prog}) is a convex optimization
and satisfies the strong duality, 
given an accessible $V_0$. 

\begin{lemma}\label{lem2}
The Program (\ref{rcp-prog}) is a convex optimization problem.
\end{lemma}
\begin{proof}
Since $R(x)$ is linear function of $x$ and (\ref{rcp-c2.1}) is linear constraint on $x$, and $\bar{\mathcal{X}}$ is convex set, the Program (\ref{rcp-prog}) is a convex optimization.
\end{proof}

\begin{lemma}\label{thm2}
If there exists a feasible allocation $x_0 \in \bar{\mathcal{X}}$, i.e., $V(x_0)\geq V_0$\footnote{The condition of Lemma \ref{thm2} can be checked by Mechanism $\mathcal{M}$ with $\alpha = 0$ which can find the maximum of GMV. On the other hand, if the threshold is so high that the maximum of GMV can not reach, then the $V_0$ will make no sense.}, The Program (\ref{rcp-prog}) satisfies strong duality property.
\end{lemma}
\begin{proof}
This lemma can be derived from the refined Slater's condition (\cite{boyd}).
\end{proof}
%
For an accessible $V_0$, strong duality holding implies that
\begin{align*}
& \underbrace{\max \limits_{x:x\in \bar{\mathcal{X}},V(x)\geq V_0}R(x)}_{\text{Primal}} = \underbrace{\min \limits_{\lambda\geq0} \max \limits_{x\in\bar{\mathcal{X}}}\big(R(x)+\lambda(V(x)-V_0)\big)}_{\text{Dual}}\\
& =\min \limits_{\lambda\geq0} \max \limits_{x\in\bar{\mathcal{X}}}\left[\int_V\sum_{i\in I} (p_i(b_i)+\lambda g_i) w_i x_i(b)f(b) db-\lambda V_0\right].
\end{align*}
Observing the inner maximization problem of the dual problem, when given a $\lambda$, it is an unconstrained problem with relaxed feasible region $\mathcal{\bar{X}}$. Since the coefficient matrix of $\bar{\mathcal{X}}$ is totally unimodular matrix (\cite{Papa}), the optimal solution of inner problem is integral form, i.e., 0-1 form. 
Therefore, the optimal solution of inner problem can be obtained by mechanism $\mathcal{M}$ with $\alpha=1/(\lambda+1)$, that is,
\begin{equation}\label{alc-rule-lambda}
x_{ik}^{\lambda}(v)=
\begin{cases}
1 & \text{if } \psi_i^{\lambda}(v_i) \text{ is the } k\text{th-highest}, k\in\{1,2,\ldots, K\},\\
0 & \text{otherwise,}
\end{cases}
\end{equation}
where $\psi_i^{\lambda}(v_i)\overset{def}{=}[(\phi_i(v_i)+\lambda g_i)w_i]/(\lambda+1)$. $p_i^{\lambda}(v)$ can be derived by replacing $x_{ik}(v)$ with $x_{ik}^{\lambda}(v)$ in (\ref{pym-rule}).

As for the outer problem, according to the complementary-slackness, the optimal $\lambda^*$ satisfies that $\lambda^*(V(x^{\lambda^*})-V_0)=0$. Hence, either $\lambda^*=0$ or $V(x^{\lambda^*})-V_0=0$ must hold. For the former case, we can testify whether $\lambda^*=0$ by running the optimal mechanism $\mathcal{M}$ in Theorem \ref{thm1} with $\alpha= 1$ and compare the output GMV with $V_0$. If the output GMV exceeds $V_0$, the optimal $\lambda^*$ is 0. Otherwise, we turn to the latter case. Because $V(x^\lambda)$ is monotone on $\lambda$ (shown by Lemma \ref{lem4}) and it only contains one variable $\lambda$, we come up with a Dichotomy, Algorithm \ref{alg:algorithm}, to output the optimal $\lambda^*$. 

\begin{lemma}\label{lem4}
$V(x^\lambda)$ is non-decreasing on $\lambda$.
\end{lemma}
\begin{proof}
Assume there are two numbers $\lambda_2 >\lambda_1\geq0$. The corresponding optimal allocations are $x^{\lambda_2}(v)$ and $x^{\lambda_1}(v)$, respectively. By the optimality, we have
\begin{align*}
& \int_V \sum_{i=1}^n w_i(\phi_i(v_i)+\lambda_1 g_i)x^{\lambda_1}_i(v)f(v)dv-\lambda_1 V_0\\
\geq & \int_V \sum_{i=1}^n w_i(\phi_i(v_i)+\lambda_1 g_i)x^{\lambda_2}_i(v)f(v)dv-\lambda_1 V_0,\\
& \int_V \sum_{i=1}^n w_i(\phi_i(v_i)+\lambda_2 g_i)x^{\lambda_2}_i(v)f(v)dv-\lambda_2 V_0\\
\geq & \int_V \sum_{i=1}^n w_i(\phi_i(v_i)+\lambda_2 g_i)x^{\lambda_1}_i(v)f(v)dv-\lambda_2 V_0.
\end{align*}
Combining the two inequalities, we have
$$(\lambda_2-\lambda_1)\int_V \sum_{i=1}^n w_i g_i(x^{\lambda_2}_i(v)-x^{\lambda_1}_i(v))f(v)dv = (\lambda_2-\lambda_1)(V(x^{\lambda_2})-V(x^{\lambda_1}))\geq0.$$
Since $\lambda_2 >\lambda_1$, we obtain $V(x^{\lambda_2})-V(x^{\lambda_1})\geq0$. It indicates the monotonicity.
\end{proof}

\begin{algorithm}
\caption{Dichotomy for Outer Problem}
\label{alg:algorithm}
\textbf{Input}: A small number $\epsilon > 0$, a large number $\lambda_{max}$ which makes $V(x^{\lambda_{max}})\geq V_0$. \\
\textbf{Output}: Optimal  $\lambda^*$.
\begin{algorithmic}[1] 
\STATE Initially, let $\lambda=0$ and run mechanism $\mathcal{M}$ with $\alpha=1$.
\IF {$V(x^0)\geq V_0$}
\STATE $\lambda^*=0$.
\ELSE
\STATE $\lambda_{min}=0$;
\WHILE{$\lambda_{max}-\lambda_{min}> \epsilon$}
\STATE Let $\lambda_{mid}= (\lambda_{max}+\lambda_{min})/2$ and run mechanism $\mathcal{M}$ with $\alpha=1/(1+\lambda_{mid})$.
\IF {$V(x^{\lambda_{mid}})\geq V_0$}
\STATE $\lambda^{max}=\lambda_{mid}$.
\ELSE
\STATE $\lambda^{min}=\lambda_{mid}$.
\ENDIF
\ENDWHILE
\ENDIF
\STATE \textbf{return} $\lambda^*$.
\end{algorithmic}
\end{algorithm}

Consequently, plugging the optimal $\lambda^*$ into $x_{ik}^{\lambda}(v)$ and $p_{ik}^{\lambda}(v)$, we derive the optimal feasible mechanism, $\mathcal{\bar{M}}=(x_i^{\lambda^*}(v), p^{\lambda^*}_i(v))$, for the relaxed constrained problem. Since $x_{ik}^{\lambda}(v)$ is 0-1 form, $\mathcal{\bar{M}}$ is the optimal mechanism for the unrelaxed constrained problem.


Owing to the strong duality, the constrained problem is equivalent to an unconstrained problem in some sense, if we choose a proper parameter. Therefore, we have the following theorem to demonstrate their relationship.
\begin{theorem}\label{thm:rlt}
The constrained problem with threshold $V_0$ is equivalent to one unconstrained problem with coefficient $\alpha^*$, where $\alpha^*= 1/(\lambda^*+1)$ and $\lambda^*$ is the optimal Lagrangian multiplier in the constrained problem.
\end{theorem}
\begin{proof}
    Since the strong duality of the relaxed constrained problem, we know that the optimal solution of primal problem is equal to the optimal solution of dual problem. On the other hand, the dual problem is the relax unconstrained problem with coefficient $\alpha^*= 1/(\lambda^*+1)$ and has the same optimal solution as the unrelaxed version, which is an unconstrained problem.
\end{proof}

%% file: Source/Ext.tex
\section{Extensions}
\label{sec:ext}
When allowing that advertisements and organic items can be arranged in a fully hybrid way, one inevitable problem is that sometimes there may be too many ad items displaying in the optimal allocations. Even though this type of allocations can achieve the optimal tradeoff, these scenarios will hurt the  e-commerce platform in the long run. In this section, we generalize the integrated system to tailor for more  practical requirements. We make two extensions: one is to restrict the number of ad items shown in one result page, the other is to add the constraint on the sparsity of ad items. In each extension, we find the optimal mechanism for the unconstrained problem (generalizing the Theorem \ref{thm1}), and build up the relationship between the constrained problem and the unconstrained problem (extending Theorem \ref{thm:rlt}).
\subsection{The Integrated Ad System with Number Budget on Ad Items}
\label{sec:ext:nb}
In this subsection, we still study designing the optimal mechanisms with an extra requirement of bounding the total number of ad items. Formally, if we require the number of ad items within $c$, we only need to add the constraint, $\sum_{i\in A}\sum_{k}x_{ik}(v)\leq c$, into $\mathcal{X}$ and denote new feasible region as $\mathcal{X'}$.

We first extend the results of unconstrained problem to this setting. From Lemma \ref{lem-IC}, \ref{lem1} and mechanism $\mathcal{M}$, we have known that items are ranked by their revised virtual values, $\psi_i(v_i)$. To abide by the new constraints, ad items appear at most $c$ times in the final layout. It is reasonable for us to take the top $c$ ad items (ordered by their highest revised virtual value) into account while ranking all items. In view of this point, we generalize the mechanism $\mathcal{M}$ to $\mathcal{M'}(c)$:

{\centering\fbox{\begin{minipage}{0.97\columnwidth}
\paragraph{Mechanism $\mathcal{M}'(c)$:}
\begin{enumerate}
\item Sort all ad items into non-increasing order by $\psi_i(v_i)$. Remain the top $c$ ad items and delete the other ad items;
\item For the remaining ad items and all organic items, run the mechanism $\mathcal{M}$ with coefficient $\alpha$.
\end{enumerate}
\end{minipage}
}\par}
Due to the optimality of mechanism $\mathcal{M}$ and satisfying the restriction on the number of ad items, we can claim that mechanism $\mathcal{M}'(c)$ is the optimal mechanism in current setting.
\begin{theorem}\label{thm:budget}
Given $\alpha\in [0,1]$ and the budget $c$, the mechanism $\mathcal{M'}(c)$ is the optimal mechanism for the unconstrained problem $UCST(\alpha)$ with the budget on the number of ad items.
\end{theorem}
\begin{proof}
    For the step 1, it remains the top-$c$ best ad items with respect to the revised virtual value, which satisfies the budget constraint. Since the optimality of mechanism $\mathcal{M}$, it outputs the optimal arrangement within $c$ ad items. 
\end{proof}

As for the constrained problem, we follow the similar logic used in Section \ref{sec:cst}. However, the main difference and difficulty is that, the coefficient matrix of $\mathcal{X}'$ is not totally unimodular anymore, compared with the previous constrained problem. 
Fortunately, even though the coefficient matrix lacks of the unimodularity, we can prove that the relaxed problem still has the 0-1 form optimal solutions by transferring it to a min-cost max-flow problem (Lemma \ref{lem5} shown in appendix). Based on the analysis in Section \ref{sec:cst} and the result of Lemma \ref{lem5}, the constrained problem still can simplify to an unconstrained problem by choosing the proper parameter, which means that we generalize the Theorem \ref{thm:rlt} into the case with budget constraint.

\subsection{The Sparse Integrated Ad System}
\label{sec:ext:sp}
In this subsection, we investigate another extension that can control the sparsity of ad items. It is motivated by the reality that some e-commerce platforms hope that ad items appear intermittently, rather than emerge together. We consider two kinds of constraints on sparsity of ad items: the sparsity on row and the sparsity on column.


\subsubsection{The Sparsity on Row}
\label{sec;ext:sp:row}
When we search the products by personal computer on e-commerce platforms, the search results are usually shown in several rows and each row contains several items. Inspired by this layout, we hope to constrain the number of ads in each row to comfort the users. In view of this point, we divide one search-result page into multiple rows and each row contains consecutive $l$ slots. Then, we require that there are at most $c$ ad items displayed in these $l$ slots. We call this model as the sparse integrated ad system -- row (SIASR). From the mathematical perspective, the above requirement can be described as a group of constraints: $\sum_{i\in A}$ $\sum_{k=ml+1}^{(m+1)l} x_{ik}(v)\leq c, \quad m=0, 1,\ldots, \lfloor K/l \rfloor$ and $\sum_{i\in A}\sum_{k=\lfloor K/l\rfloor l+1}^{K} x_{ik}(v)\leq c$. Add them to $\mathcal{X}$ and denote the new constraint feasible set as $\mathcal{\tilde{X}}$. 

We begin with the unconstrained problem in this case. Observing these new constraints, we can find that each one is actually a constraint on the number of ad items for $l$ slots. In another word, if we divide one search-result page into $\lfloor K/l \rfloor+1$ parts, every part is an IAS with number budget on ad items (introduced in Section \ref{sec:ext:nb}). Consequently, We can run mechanism $\mathcal{M}'(c)$ for each part to output a feasible arrangement for SIASR. 
Since the optimality of $\mathcal{M}'(c)$, every part achieves the optimal arrangement, which induces the global optimal arrangement for all slots. Therefore, we derive the optimal mechanism in this setting, defined as $\mathcal{\tilde{M}}(c,l)$.

{\centering\fbox{\begin{minipage}{0.97\columnwidth}
\paragraph{Mechanism $\mathcal{\tilde{M}}(c, l)$:}
\begin{enumerate}
\item Run mechanism $\mathcal{M}'(c)$ to allocate the first $l$ slots on all items;
\item Run mechanism $\mathcal{M}'(c)$ to allocate the next $l$ slots on all remaining items;
\item Repeat step 2 until all slots are allocated. (In the last round, the number of slots may be less than $l$).
\end{enumerate}
\end{minipage}
}\par}

\begin{theorem}\label{thm:SIASR}
Given $\alpha\in [0,1]$, the number of items in a row, $l$ and the budget $c$ , the mechanism $\mathcal{\tilde{M}}(c, l)$ is the optimal mechanism for the unconstrained problem $UCST(\alpha)$ of SIASR.
\end{theorem}
\begin{proof}
    Since the optimality of mechanism $\mathcal{M}'(c)$, for consecutive $l$ slots, it outputs the optimal arrangement within $c$ ad items. On the other hand, the step 2 assures that it achieves the current optimal layout with budget constraint among the different groups. Therefore, the allocation of mechanism $\mathcal{\tilde{M}}(c, l)$ is the global optimal.
\end{proof}

For the constrained problem, guiding from the method in Section \ref{sec:cst}, the remaining question is to show the 0-1 form of the optimal solution of the relaxed problem. Compared to subsection \ref{sec:ext:nb}, the problem in this subsection is trickier, since multiple parts enhance the difficulty on construction of the equivalent min-cost max-flow problem. However, we overcome it by Lemma \ref{lem6} (shown in appendix). So far, we solve the constrained problem of SIASR and extend the Theorem \ref{thm:rlt} to this setting.

\subsubsection{The Sparsity on Column}

In this extension, we would like to investigate the case in terminal of mobile devices. Generally, after inputting a keyword, the search results are displayed in one column (not the same as the layout of terminal of personal computer). If we still want to restrict the sparsity of ads, we need to add the constraints on any consecutive $l$ slots, i.e., there are at most $c$ ad items displayed in any consecutive $l$ slots. Mathematically, add a series of constraints, $\sum_{i\in A}\sum_{k=m+1}^{m+l} x_{ik}(v)\leq c, m=0, 1,\ldots, K-l$, into $\mathcal{X}$ and denote the new feasible region as $\mathcal{\hat{X}}$. We call this model as the sparse integrated ad system -- column (SIASC).

We focus on the unconstrained problem first. Note that, compared to SIASR, SIASC owns more sparsity constraints on consecutive slots. However, for the first $l$ slots, the constraint is the same. Intuitively, the allocation for first slots should also be the same. After deciding the first $l$ slots and observing the next constraint on slot $2$ to slot $l+1$, it only need to decide which item (ad item or organic item) should be allocated to slot $l+1$ to meet the sparsity constraint. In this way, we can continue filling the following slots by checking the criteria of sparsity. Obeying this rule, we propose the mechanism $\mathcal{\hat{M}}(c, l)$.

{\centering\fbox{\begin{minipage}{0.97\columnwidth}
\paragraph{Mechanism $\mathcal{\hat{M}}(c, l)$:}
\begin{enumerate}
\item Run mechanism $\mathcal{M}'(c)$ to allocate the first $l$ slots on all items;
\item Run mechanism $\mathcal{M}'(c-a)$ to allocate the next one slot on all remaining items, where $a$ is the number of ad items in the last $l-1$ slots;
\item Repeat step 2 until all slots are allocated.
\end{enumerate}
\end{minipage}
}\par}

\begin{theorem}\label{thm:SIASC}
Given $\alpha\in [0,1]$, the number of items in a row, $l$ and the budget $c$ , the mechanism $\mathcal{\hat{M}}(c, l)$ is the optimal mechanism for the unconstrained problem $UCST(\alpha)$ of SIASC.
\end{theorem}
\begin{proof}
On the one hand, since the optimality of mechanism $\mathcal{M}'(c)$, for consecutive $l$ slots, it outputs the optimal arrangement within $c$ ad items. On the other hand, the step 2 of mechanism $\mathcal{\hat{M}}(c, l)$ assures the budget constraint among the different groups. Therefore, the allocation of mechanism $\mathcal{\hat{M}}(c, l)$ is the global optimal.
\end{proof}

As for the constrained problem, we can extend the method used in SIASR into this setting. The unique difference is that, in this setting , there are more constraints on consecutive $l$ slots. However, it does not affect the difficulty eventually. We only need to modify the structure of Figure \ref{network2} slightly: add a mount of nodes, $k_i$ to represent more groups and link the arcs  $s^{''}_{j}k_i, j\in\{i, i+1, \cdots, i+l-1\}$ and the arcs $k_i t$. The other analysis is similar. In this way, we generalize the Theorem \ref{thm:rlt} to the case of SIASC.

%% file: Source/Exp.tex
\section{Experiments}
\label{sec:exp}
In this section, we will verify our mechanisms with four experiments. 
Our experimental environment is that we select 356 different keywords randomly, where each keyword contains at least 400 and at most 2000 candidate items (including ad items and organic items). For each keyword, we simulate the first result page with 20 slots. We find that the values of ad items are fitted into a lognormal distribution, and assume that quality factor $w_i=1$ for experiment 1-3, but has a distribution in experiment 4. The estimated volumes of items can be extracted from the real data, directly. Because of the diversity of the keywords and to avoid speciality, we repeat experiments on each keyword and aggregate them into normalized results (figures).

\subsection{The Monotonicity of Revenue and Volume}
The first experiment aims to testify the availability of mechanism $\mathcal{M}$. We exploit eleven different values of the weighted coefficient $\alpha$ from 0 to 1 with equal gap 0.1. When $\alpha$ is determined, for each keyword, we generate a group of random bids by the lognormal distribution
\footnote{The number of generating bids is the number of advertisers bidding for this keyword in real data. For different keywords, we may learn the different lognormal distributions.}
learned from the real data and run mechanism $\mathcal{M}$, recording the corresponding revenue and GMV. We repeat this experiment for 5000 times, and take mean of the revenues and GMVs as the expected revenue and the expected GMV of this keyword. Aggregating the results of all keywords, we bring forth the final result shown in Figure \ref{exp1}.

\begin{figure}[htbp]
    \centering
    \includegraphics[width=0.5\textwidth]{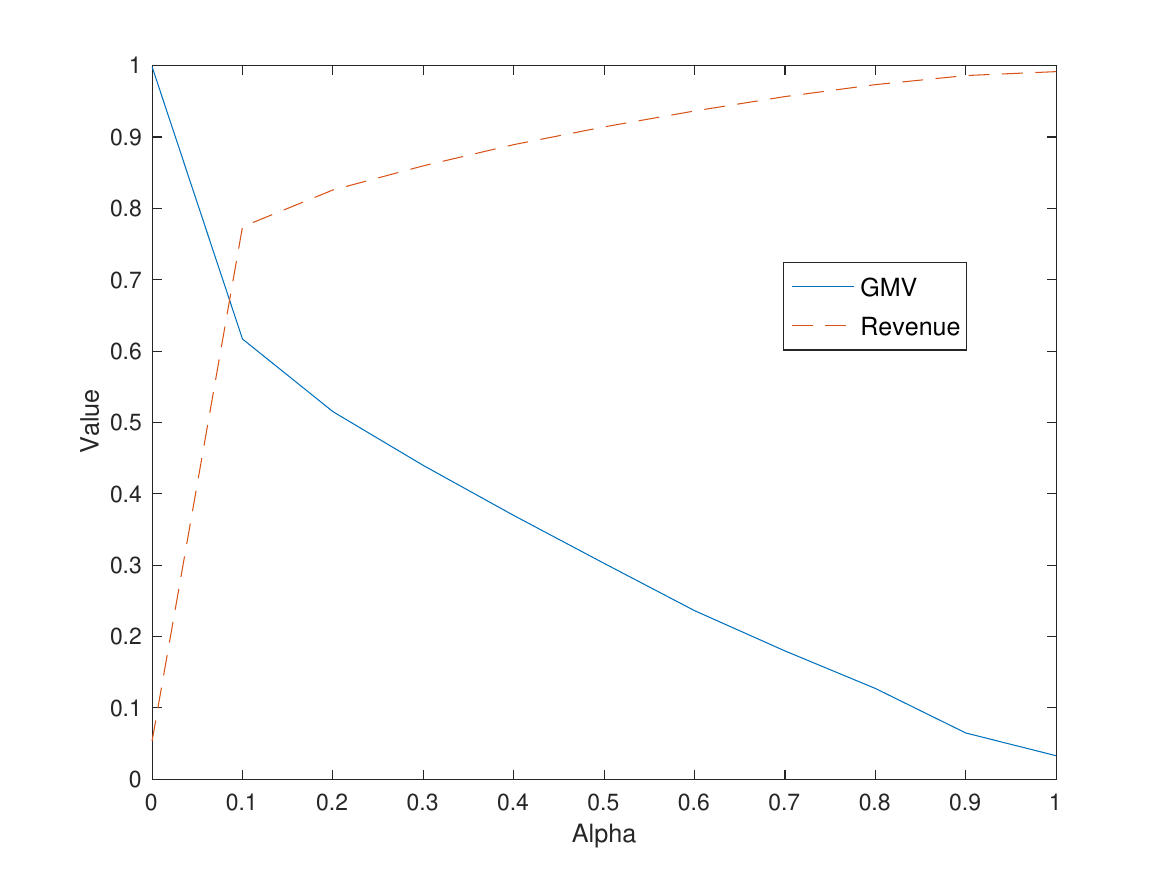}
    \caption{Expected revenue and GMV under different $\alpha$}
    \label{exp1}
\end{figure}

Observing from Figure \ref{exp1}, we can see that as $\alpha$ increases from 0 to 1, revenue improves while GMV decreases, which fulfills our intuition that the two objectives are in conflict. Hence, if platforms focus more on revenue (or GMV), they only need to increase (decrease) the parameter $\alpha$ of mechanism $\mathcal{M}$. On the other hand, for each $\alpha$, how much fraction of the optimal value the mechanism $\mathcal{M}$ achieves can be found from Figure \ref{exp1}. This can help the platform choose the appropriate $\alpha$ to balance the revenue and volume in practice.

\subsection{The Execution of Theorem \ref{thm:rlt}}
The second experiment aims to verify Theorem \ref{thm:rlt}. For each keyword, we give eleven different threshold of GMV (we pick up them from 0 to the maximum feasible GMV with equal gap). For each threshold, we execute Algorithm \ref{alg:algorithm} to derive the optimal $\lambda^*$ and compute the corresponding $\alpha^*$. During executing Algorithm \ref{alg:algorithm}, since $V(x)$ is a calculus, we always repeat generating bids for 500 times and take mean of GMVs as the value of $V(x)$. After testing all the thresholds, we can depict the relationship curve of $V_0$ and $\alpha$, shown in Figure \ref{exp2}.

\begin{figure}[htbp]
    \centering
    \includegraphics[width=0.5\textwidth]{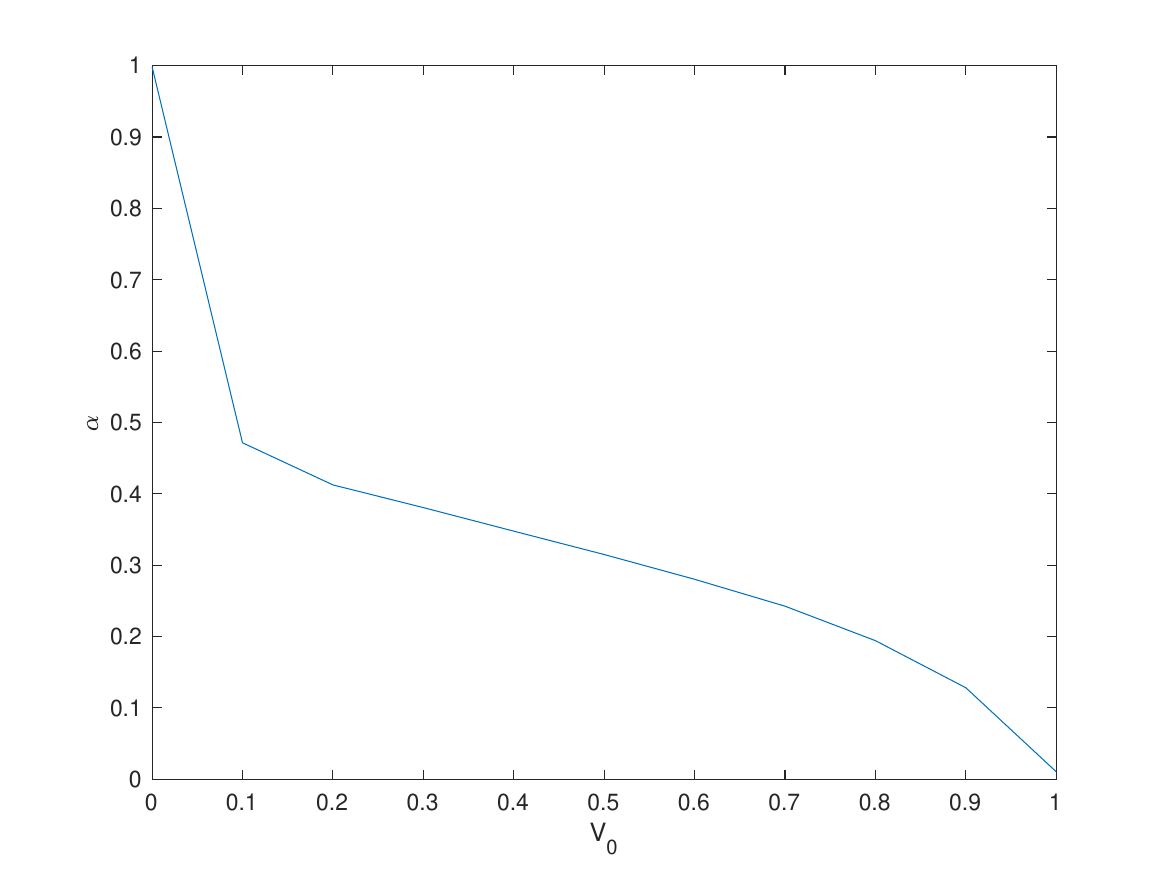}
    \caption{The relationship between $V_0$ and $\alpha$}
    \label{exp2}
\end{figure}

Figure \ref{exp2} demonstrates that the larger of the volume threshold (i.e., $V_0$) is, the smaller $\alpha$ will be. This result is coincident with experiment \ref{exp1}. The most significant point is that, for constrained problem, it can be transferred into an unconstrained problem, which reduces the difficulty of operation. Experiment \ref{exp2} also shows the correctness of Theorem \ref{thm:rlt} and operability of our mechanism.

\subsection{Priority of Our Mechanisms}
In this subsection, we show the superiority of our mechanisms by two experiments. In each experiment, we compare our mechanisms with the used mechanisms of e-commerce platforms. The difference of two experiments is that, in second one, we take the correlation between the weight and the value of advertisers into consideration. We start with the uncorrelated case.

\subsubsection{Uncorrelated Case}
In the third experiment, we want to compare our optimal mechanism with the commonly used mechanism (i.e., only taking out some particular topmost slots to ad items). First, the current mechanism can be described as: we fix several topmost slots, and run Myerson optimal auction\footnote{Note that the platforms usually run generalized second-price(GSP) auction. \cite{Ostr} showed that when all advertisers have identical regular distributions of values, the expected revenue of GSP with reserve price $r^*$ such that $\phi(r^*)=0$ in the bidder-optimal envy-free equilibrium is the same as the expected revenue of Myerson optimal auction.} on ad items. For organic part, we allocate organic items to the remaining slots by their volumes sorted in non-increasing order. In this way, we can compute the expected revenue and GMV in the traditional layout. Then, we use the just obtained GMV as the threshold $V_0$ and run mechanism $\mathcal{\bar{M}}$ (i.e., we solve the constrained problem with $V_0$). We derive another expected revenue from our mechanism. Repeat the above procedure, as the number of fixed slots changes from 1 to 8. Finally, we compare the expected revenue gained in two mechanisms and illustrate them in Figure \ref{exp3}.

\begin{figure}[htbp]
    \centering
    \includegraphics[width=0.5\textwidth]{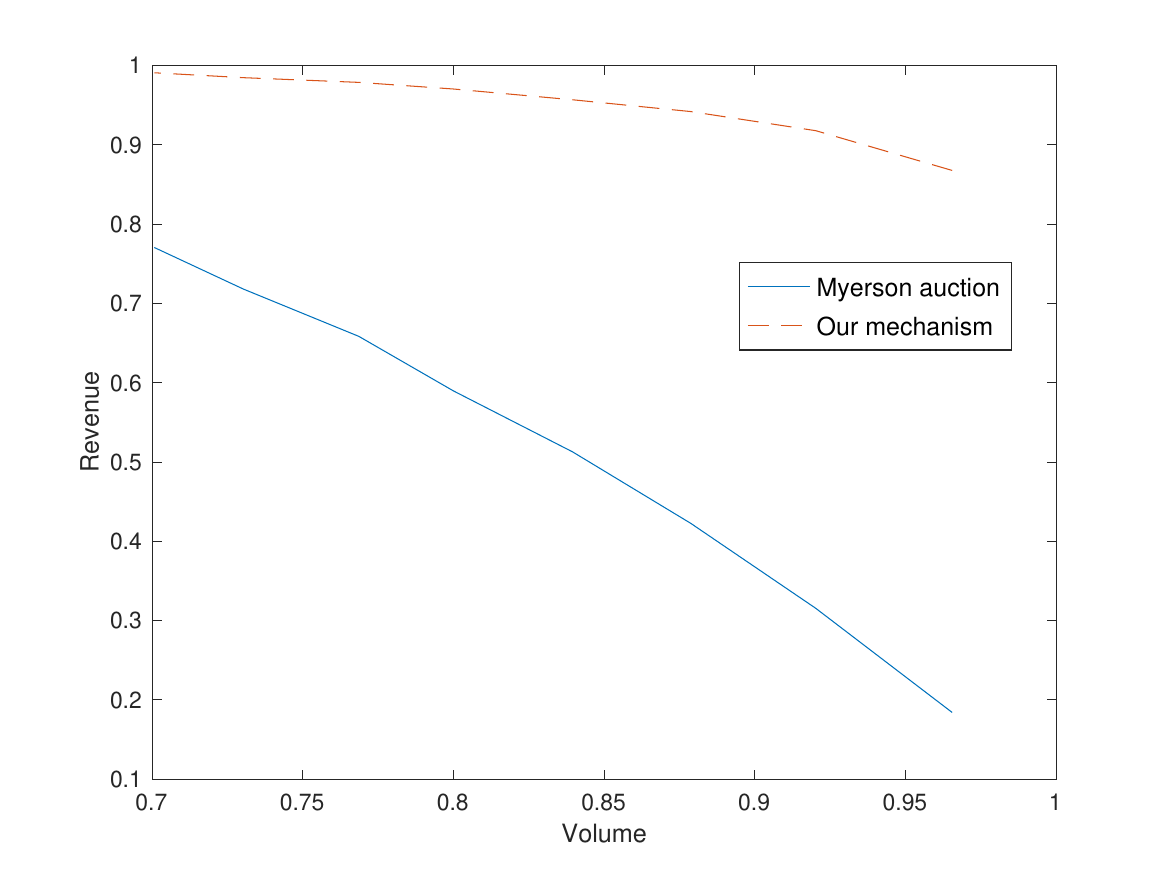}
    \caption{Comparison between two mechanisms}
    \label{exp3}
\end{figure}

From Figure \ref{exp3} we can see, on the premise of reaching the same GMV, the revenue curve of our mechanism lies on the top of that of Myerson optimal mechanism. The main reason is that, the conventional layout always maximize revenue in particular slots and then maximize GMV in the rest slots, which are separate two steps, while we maximize revenue and guarantee GMV, simultaneously. Therefore, IAS is more flexible, and can keep the same GMV but gain more revenue. This also reflects the superiority of mixture arrangements.

\subsubsection{Correlated Case}
The last experiment is based on the third experiment. In this experiment, we assume that the quality factor, $w_i$, follows a distribution (we learn that it is fitted into a beta distribution). We wonder that, if there exists a correlation between $w_i$ and $v_i$, what the results of experiment 3 will be. Five different correlation coefficients, -1, -0.5, 0, 0.5, 1, are studied. For each correlation coefficient, we repeat the experiment 3 and depict revenue curves. We union five pictures into one figure shown in Figure \ref{fig:corr}.

\begin{figure*}[htbp]
\centering
\subfigure[r=-1]{
\label{corr:1}
\begin{minipage}[t]{0.5\textwidth}
\centering
\includegraphics[width=\textwidth]{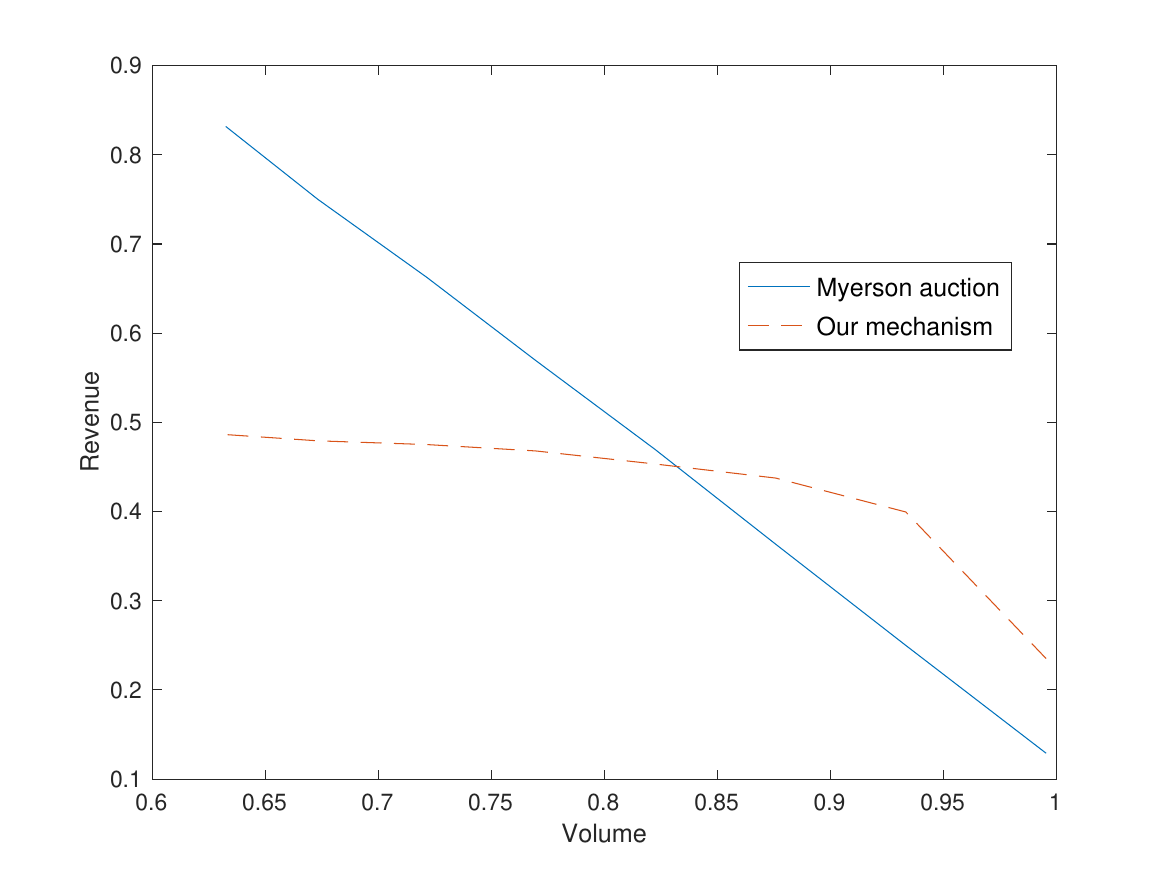}
\end{minipage}%
}%
\subfigure[r=-0.5]{
\begin{minipage}[t]{0.5\textwidth}
\centering
\includegraphics[width=\textwidth]{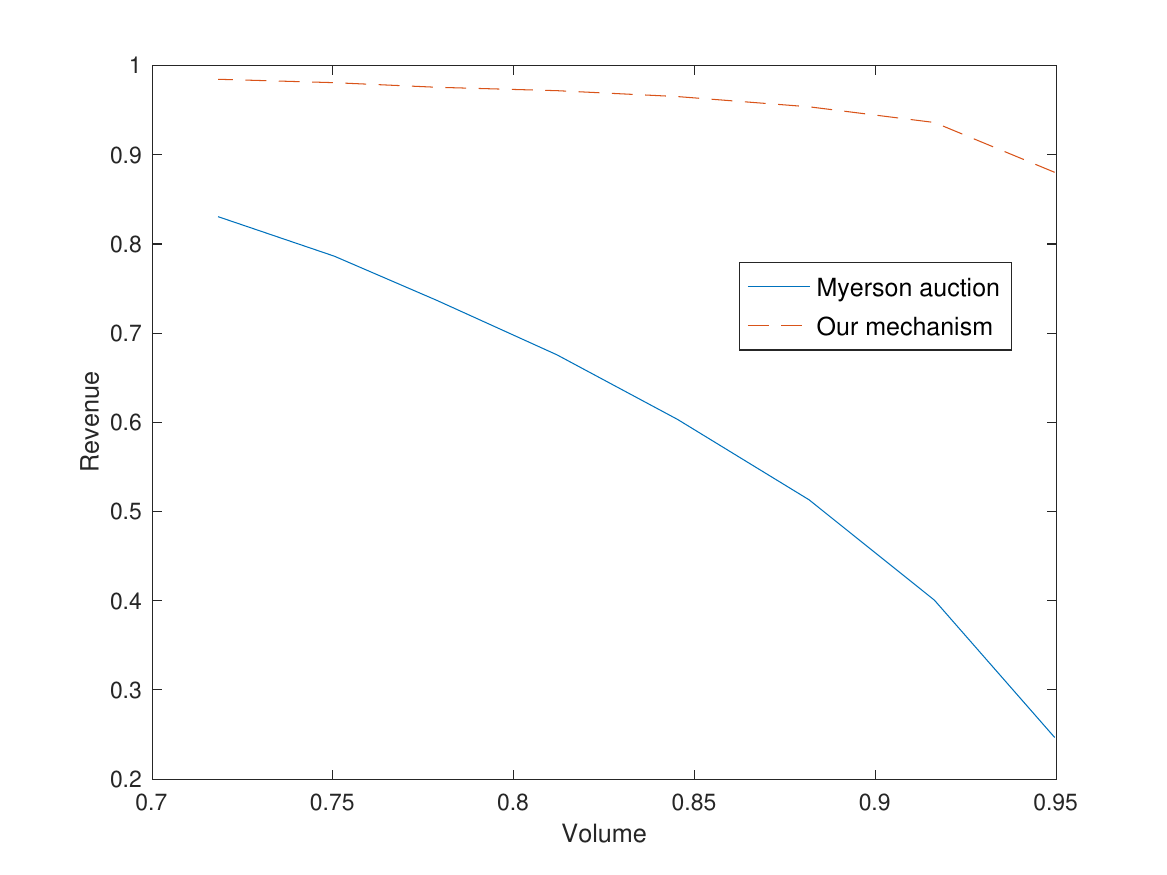}
\end{minipage}%
}%
\\
\subfigure[r=0]{
\begin{minipage}[t]{0.5\textwidth}
\centering
\includegraphics[width=\textwidth]{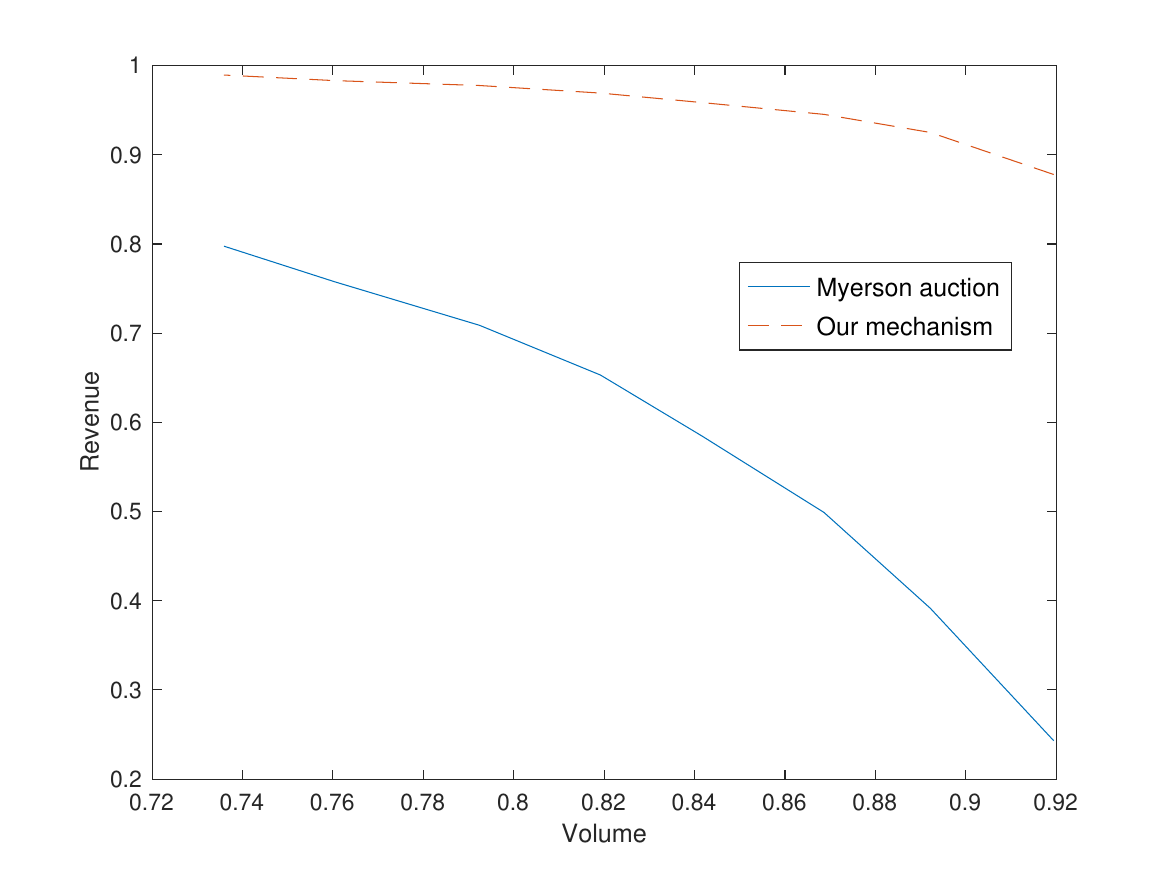}
\end{minipage}
}%
\subfigure[r=0.5]{
\begin{minipage}[t]{0.5\textwidth}
\centering
\includegraphics[width=\textwidth]{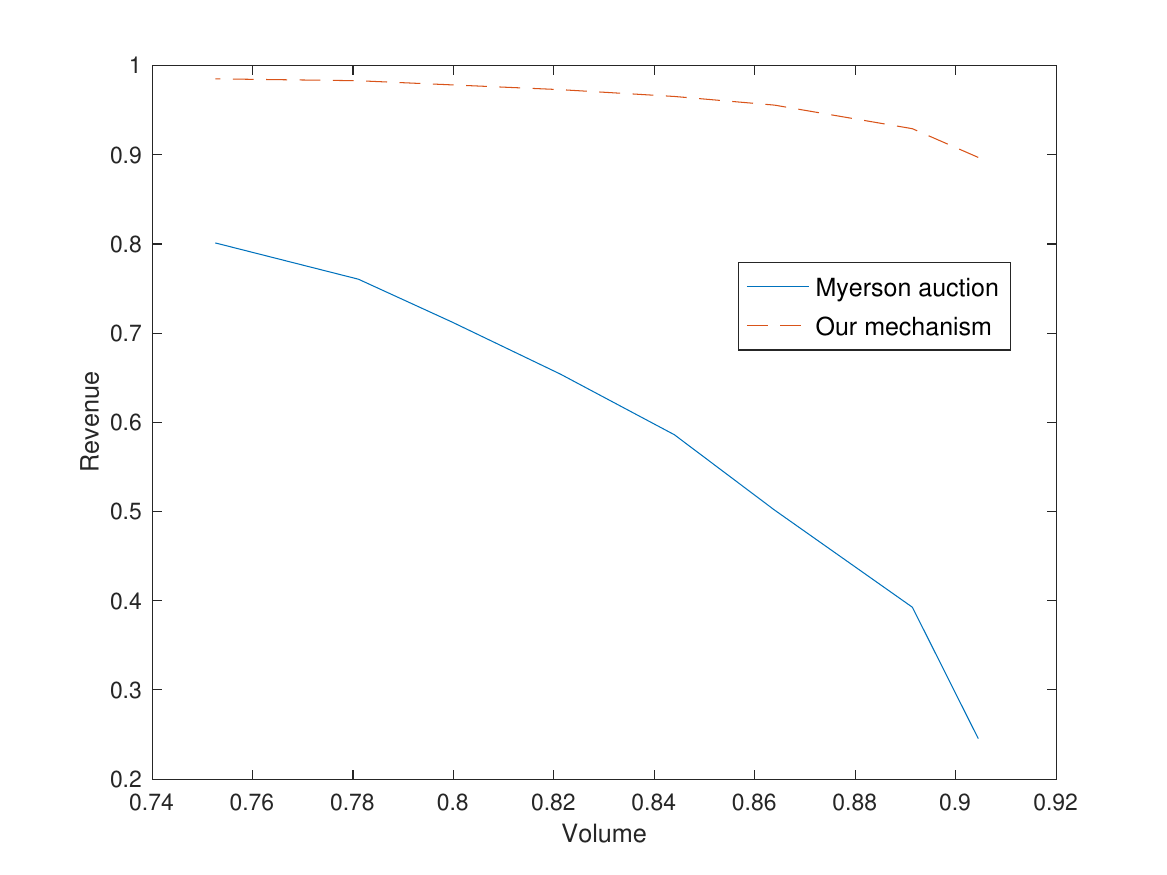}
\end{minipage}%
}%
\\
\subfigure[r=1]{
\begin{minipage}[t]{0.5\textwidth}
\centering
\includegraphics[width=\textwidth]{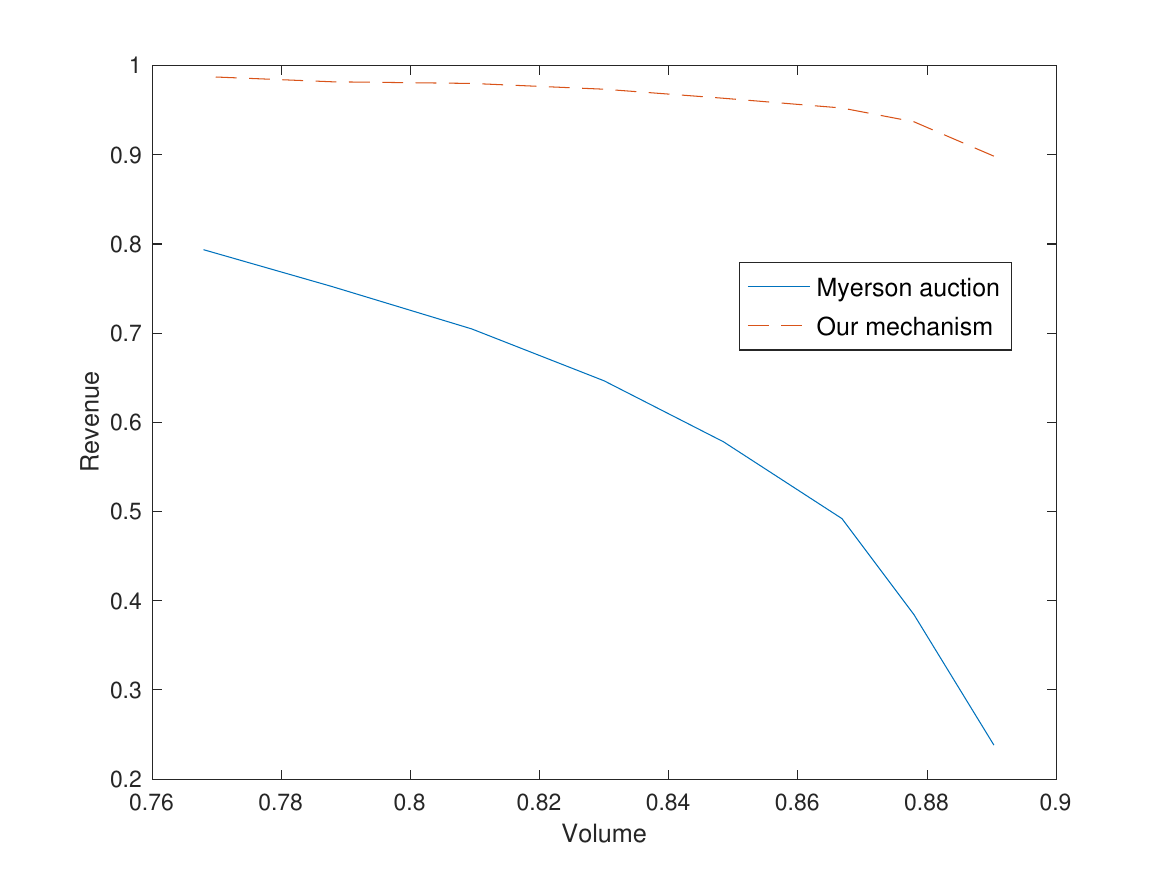}
\end{minipage}%
}%
\centering
\caption{Correlation experiments}
\label{fig:corr}
\end{figure*}

Figure \ref{fig:corr} illustrates that, except the condition of totally negative correlation, other conditions present the similar trend as Figure \ref{exp3}. However, when correlation coefficient changes from 1 to -0.5, the largest revenue gap becomes smaller. As for Figure \ref{corr:1}, it can be explained as that, when the quality factor and value have a negative correlation, the ad items with higher value contribute less GMV, because their quality factor are relatively smaller, which leads the fact that the $w_ig_i$ is smaller. To achieve the same GMV, our mechanism may sacrifice some revenue. In brief, except some extreme conditions, our mechanism show the advantage even if there exists the correlation between $w_i$ and $v_i$.

%% file: Source/Appendix.tex
\appendix
\section*{Appendix}


\section{Necessary Example and Tables}

\begin{example}
We consider the mainline of one search-result page setting with 10 slots for items. The click-through-rates (CTRs) of these 10 slots are from 0.1 to 1 with the same gap. There are total 10 candidates, 3 sponsored advertisements and 7 organic items
. The detailed numbers of bids and volumes are shown in Appendix, Table \ref{GSP}, where the ids of items are ranked by the volume.

In Table \ref{GSP}, we fix the top 3 slots for the ads and run the GSP mechanism, that is, if an advertiser wins, she only need to pay the lowest price that can keep the current slot. Then we rank the rest organic items by GMV, obtaining the outcome of the conventional pattern.  The GMV and revenue in this setting are 451.3 and 21.9.

\begin{table}[ht!]
\begin{center}
\scalebox{0.95}{
    \begin{tabular}{ | l | l | l | l | p{1.5cm} | l | l |p{1.3cm}|p{1.3cm}|}
    \hline

Id &Type &Volume &Bid &Rank by 0.5Bid + 0.5Volume &Payment &CTR &Actual GMV	&Actual Payment\\ \hline
A1	&ad&70&15&15&12&1&70&12\\ \hline
A2	&ad&75	&12	&12	&11	&0.9	&67.5	&9.9\\ \hline
A3	&ad&90	&11	&11	&0	&0.8	&72	&0\\ \hline

O1 	&organic &100 &0 &0 &&0.7 &70 &0\\ \hline
O2	&organic&90&0 &0&&0.6&54&0\\ \hline
O3	&organic	&85	&0	&0	&&0.5	&42.5	&0\\ \hline
O4	&organic	&80	&0	&0	&&0.4	&32	&0\\ \hline
O5	&organic	&75	&0	&0	&&0.3	&22.5	&0\\ \hline
O6	&organic	&70	&0	&0	&&0.2	&14	&0\\ \hline
O7	&organic	&68	&0	&0	&&0.1	&6.8	&0\\ \hline
\multicolumn{7}{|l|}{Total}&451.3&21.9\\ \hline
\end{tabular}
}
\caption{GSP mechanism}
\label{GSP}
\end{center}
\end{table}

Now, we regard organic results and advertisements as an entirety and rank them at the same time.  Think over a new allocation method: rank by 0.5bid + 0.5volume for all items and for advertisements, still charge the lowest price that keeps the current slot. The result is shown in Appendix, Table \ref{ITG}. The GMV and revenue are enhanced to 465.8 and 22.

\begin{table}[ht!]
\begin{center}
\scalebox{0.95}{
    \begin{tabular}{ | l | l | l | l | p{1.5cm} | l | l |p{1.3cm}|p{1.3cm}|}
    \hline

Id &Type &Volume &Bid &Rank by 0.5Bid + 0.5Volume &Payment &CTR &Actual GMV	&Actual Payment\\ \hline
A3	&ad	&90	&11	&50.5	&10	&1	&90	&10\\ \hline
O1	&organic	&100	&0	&50	&	&0.9	&90	&0\\ \hline
O2	&organic	&90	&0	&45&		&0.8	&72	&0\\ \hline
A2	&ad	&75	&12	&43.5	&10	&0.7	&52.5	&7\\ \hline
O3	&organic	&85	&0	&42.5&		&0.6	&51	&0\\ \hline
A1	&ad	&70	&15	&42.5	&10	&0.5	&35	&5\\ \hline
O4	&organic	&80	&0	&40	&	&0.4	&32	&0\\ \hline
O5	&organic	&75	&0	&37.5&		&0.3	&22.5	&0\\ \hline
O6	&organic	&70	&0	&35	&	&0.2	&14	&0\\ \hline
O7	&organic	&68	&0	&34	&	&0.1	&6.8	&0\\ \hline
\multicolumn{7}{| l |}{Total} &465.8	&22\\ \hline
\end{tabular}
}
\caption{Integrated Mechanism}
\label{ITG}
\end{center}
\end{table}

We can see the layout of mixed arrangement presented by Table \ref{ITG} performs better on both GMV and revenue. This answers the second question and shows that the mixed arrangement could benefit.
\end{example}

\section{Missing Proofs of Section \ref{sec:ext}}
\subsection{Proof of Lemma \ref{lem5}}
\begin{lemma}\label{lem5}
For the IAS with budget on ad items, the relaxed inner problem of constrained problem can be constructed as a min-cost max-flow problem, and has a 0-1 form optimal solution.
\end{lemma}
\begin{proof}
This problem, indeed, can be reduced to a min-cost max-flow problem, shown by Figure \ref{network}.

\begin{figure}[htbp]
  \centering
  \includegraphics[width=0.5\textwidth]{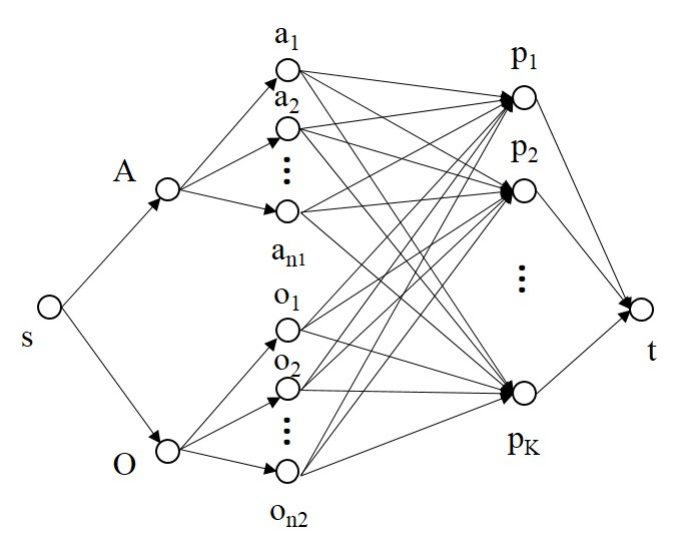}
  \caption{Min-cost max-flow network}
  \label{network}
\end{figure}

In Figure \ref{network}, $a_i$ and $o_j$ represent ad items and organic items, respectively. $p_k$ represents the slots. The capacities of directed arcs from left to right are: the capacity of $sA$ is $c$, and the capacity of $sO$ is infinity. For $Aa_i$, $Oo_j$, $a_ip_k$, $o_jp_k$, $\forall i,j,k$, the capacity is 1. As for $p_kt$, $\forall k$, its capacity is  1.

The cost of directed arcs are:  for $a_ip_k$ or $o_jp_k$, $\forall, i,j,k$, the cost of unit flow is $-w_i(\phi_i(v_i)+\lambda g_i)\beta_k$ or $-w_j\lambda g_j\beta_k$. For other arcs, the cost is 0.

We can check that the total cost generating from above flow is equal to the value of objective function of inner problem. The conservation of flows are exactly the constraints in $\mathcal{X}'$. Since this min-cost max-flow problem has a 0-1 form optimal solution, the optimal solution of relaxed inner problem is also 0-1 form.
\end{proof}

\subsection{Proof of Lemma \ref{lem6}}

\begin{lemma}\label{lem6}
For the SIASR, the relaxed inner problem of constrained problem can be constructed as a min-cost max-flow problem, and has a 0-1 form optimal solution.
\end{lemma}

\begin{proof}
We still reduce this problem into a min-cost max-flow problem, shown in Figure \ref{network2}. In Figure \ref{network2}, $a_i$ and $o_j$ represent ad items and organic items, respectively. $s'_k$ and $s''_k$ represent the slots. $k_i$ is the node to restrict the number of ads in group $i$. $s_k$ is the node to restrict that one slot only has one item.

\begin{figure}[htbp]
  \centering
  \includegraphics[width=0.5\textwidth]{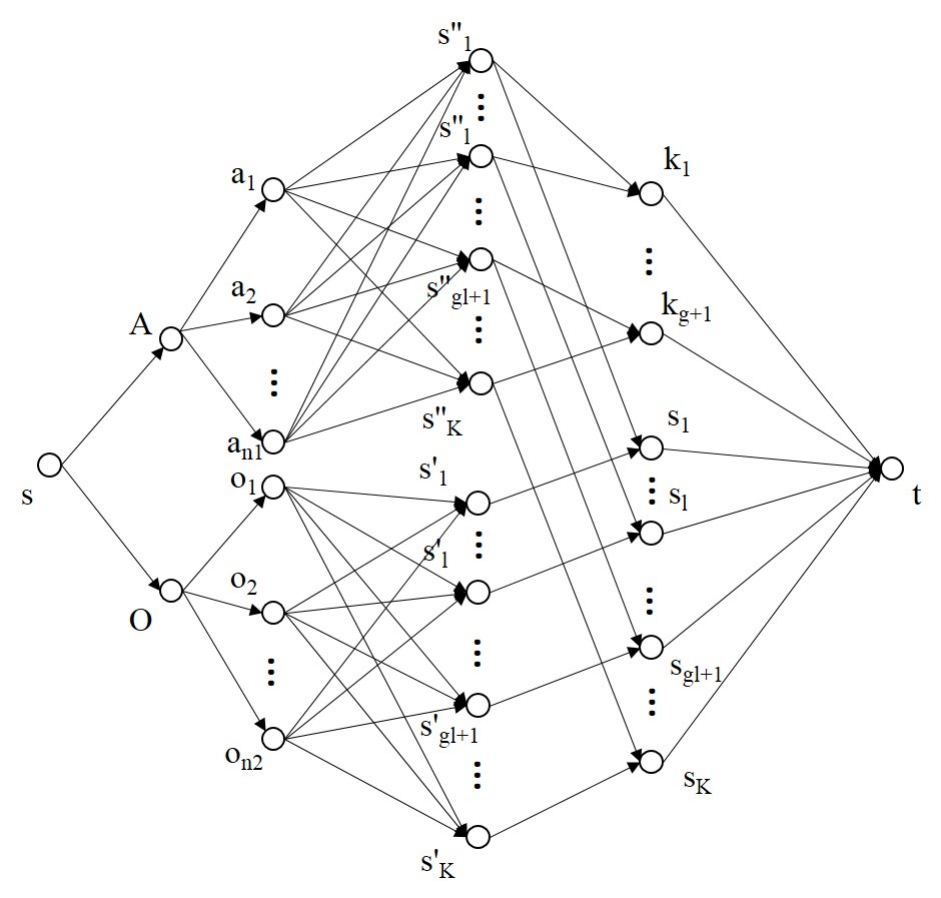}
  \caption{Min-cost max-flow network 2}
  \label{network2}
\end{figure}

The capacities of some certain directed arcs (from right to left) are: for $k_it$, $\forall i$, the capacity is $c$; 
for $Aa_i$ and $a_is''_k$, $\forall i,k$, the capacity is 2; for $sA$ and $sO$, the capacity is infinity. The capacity of the others is 1.

The cost of certain directed arcs are: for $a_is''_k$, $\forall i,k$, the cost is $-(w_i(\phi_i(v_i)+\lambda g_i)\beta_k)/2$; for $o_js'_k$, $\forall, i,j,k$, the cost is $-w_j\lambda g_j\beta_k$. For other arcs, the cost is 0.

Note that the feasible region of this network flow problem is a little broader than the relaxed inner problem. This is because the flow from $s''_j$ to $s_k$ and the flow from $s''_j$ to $k_i$ can be different in a feasible flow, but in the constraints of relaxed inner problem, they must be the same. However, we can guarantee that, in the optimal solution, the value of these two flow are equal. Since the capacity of in-arc of $s''_j$ is 2 and the capacity of out-arc is 1, in the optimal solution, the flow passing the $s''_j$ must be split equally. It means that the optimal solution of this min-cost max-flow problem can be transferred to a 0-1 form feasible solution to the relaxed inner problem. Because the total cost generating from above flow is equal to the value of objective function of relaxed inner problem, this corresponding 0-1 form solution is the optimal solution of the relaxed inner problem.
\end{proof}

%% file: main.bbl
\begin{thebibliography}{10}
\providecommand{\url}[1]{\texttt{#1}}
\providecommand{\urlprefix}{URL }
\providecommand{\doi}[1]{https://doi.org/#1}

\bibitem{Abr}
Abrams, Z., Schwarz, M.: Ad auction design and user experience. In: Deng, X.,
  Graham, F.C. (eds.) Internet and Network Economics. pp. 529--534. Springer
  Berlin Heidelberg, Berlin, Heidelberg (2007)

\bibitem{GAR06}
Aggarwal, G., Goel, A., Motwani, R.: Truthful auctions for pricing search
  keywords. In: Proceedings of the 7th ACM conference on Electronic commerce.
  pp.~1--7. ACM (2006)

\bibitem{Ath}
Athey, S., Ellison, G.: Position auctions with consumer search. The Quarterly
  Journal of Economics  \textbf{126}(3),  1213--1270 (2011)

\bibitem{AU06}
Austin, D.: How google finds your needle in the web’s haystack. American
  Mathematical Society Feature Column  \textbf{10}(12) (2006)

\bibitem{AL04}
Avrachenkov, K., Litvak, N.: Decomposition of the google pagerank and optimal
  linking strategy  (2004)

\bibitem{Bac}
Bachrach, Y., Ceppi, S., Kash, I.A., Key, P., Kurokawa, D.: Optimising
  trade-offs among stakeholders in ad auctions. In: Proceedings of the
  fifteenth ACM conference on Economics and computation. pp. 75--92. ACM (2014)

\bibitem{boyd}
Boyd, S., Vandenberghe, L.: Convex optimization. Cambridge university press
  (2004)

\bibitem{CNZ17}
Chu, L.Y., Nazerzadeh, H., Zhang, H.: Position ranking and auctions for online
  marketplaces. Available at SSRN 2926176  (2017)

\bibitem{CLK71}
Clarke, E.H.: Multipart pricing of public goods. Public choice  \textbf{11}(1),
   17--33 (1971)

\bibitem{CRW07}
Crowcroft, J.: Net neutrality: the technical side of the debate: a white paper.
  ACM SIGCOMM Computer Communication Review  \textbf{37}(1),  49--56 (2007)

\bibitem{EL11}
Edelman, B., Lockwood, B.: Measuring bias in ‘organic’web search.
  Unpublished manuscript  (2011)

\bibitem{EOS07}
Edelman, B., Ostrovsky, M., Schwarz, M.: Internet advertising and the
  generalized second-price auction: Selling billions of dollars worth of
  keywords. American economic review  \textbf{97}(1),  242--259 (2007)

\bibitem{GRV73}
Groves, T., et~al.: Incentives in teams. Econometrica  \textbf{41}(4),
  617--631 (1973)

\bibitem{Lah}
Lahaie, S., Pennock, D.M.: Revenue analysis of a family of ranking rules for
  keyword auctions. In: Proceedings of the 8th ACM conference on Electronic
  commerce. pp. 50--56. ACM (2007)

\bibitem{Li}
Li, J., Liu, D., Liu, S.: Optimal keyword auctions for optimal user
  experiences. Decision Support Systems  \textbf{56},  450--461 (2013)

\bibitem{Lik}
Likhodedov, A., Sandholm, T.: Auction mechanism for optimally trading off
  revenue and efficiency. In: Proceedings of the 4th ACM conference on
  Electronic commerce. pp. 212--213. ACM (2003)

\bibitem{LMST17}
L’Ecuyer, P., Maill{\'e}, P., Stier-Moses, N.E., Tuffin, B.:
  Revenue-maximizing rankings for online platforms with quality-sensitive
  consumers. Operations Research  \textbf{65}(2),  408--423 (2017)

\bibitem{MT14}
Maill{\'e}, P., Tuffin, B.: Telecommunication network economics: from theory to
  applications. Cambridge University Press (2014)

\bibitem{Mask}
Maskin, E., Riley, J.: Optimal multi-unit auctions'. International library of
  critical writings in economics  \textbf{113},  5--29 (2000)

\bibitem{Myerson}
Myerson, R.B.: Optimal auction design. Mathematics of operations research
  \textbf{6}(1),  58--73 (1981)

\bibitem{Ostr}
Ostrovsky, M., Schwarz, M.: Reserve prices in internet advertising auctions: A
  field experiment. In: Proceedings of the 12th ACM conference on Electronic
  commerce. pp. 59--60. ACM (2011)

\bibitem{Papa}
Papadimitriou, C.H., Steiglitz, K.: Combinatorial optimization: algorithms and
  complexity. Courier Corporation (1998)

\bibitem{Rob}
Roberts, B., Gunawardena, D., Kash, I.A., Key, P.: Ranking and tradeoffs in
  sponsored search auctions. ACM Transactions on Economics and Computation
  (TEAC)  \textbf{4}(3), ~17 (2016)

\bibitem{Shen}
Shen, W., Tang, P.: Practical versus optimal mechanisms. In: Proceedings of the
  16th Conference on Autonomous Agents and MultiAgent Systems. pp. 78--86.
  International Foundation for Autonomous Agents and Multiagent Systems (2017)

\bibitem{Sun}
Sundararajan, M., Talgam-Cohen, I.: Prediction and welfare in ad auctions.
  Theory of Computing Systems  \textbf{59}(4),  664--682 (2016)

\bibitem{Thom}
Thompson, D.R., Leyton-Brown, K.: Revenue optimization in the generalized
  second-price auction. In: Proceedings of the fourteenth ACM conference on
  Electronic commerce. pp. 837--852. ACM (2013)

\bibitem{VR07}
Varian, H.R.: Position auctions. international Journal of industrial
  Organization  \textbf{25}(6),  1163--1178 (2007)

\bibitem{VK61}
Vickrey, W.: Counterspeculation, auctions, and competitive sealed tenders. The
  Journal of finance  \textbf{16}(1),  8--37 (1961)

\bibitem{WH10}
Williams, H.: Measuring search relevance  (2010)

\bibitem{WR12}
Wright, J.D.: Defining and measuring search bias: Some preliminary evidence.
  International center for law \& economics, November pp. 12--14 (2011)

\end{thebibliography}
